\documentclass[a4paper,12pt]{article}
\usepackage[leqno]{amsmath}
\usepackage{color}
\usepackage{amsthm}
\usepackage{amssymb}
\usepackage{graphicx}
\usepackage{subfig}
\usepackage{hyperref}
\usepackage{enumerate}
\usepackage{mathrsfs}
\usepackage[authoryear]{natbib}

\newtheorem{proposition}{Proposition}

\newtheorem{lemma}{Lemma}

\newtheorem{remark}{Remark}

\newcommand{\diff}{\mathrm{d}}
\newcommand{\op}[1]{\mathcal{#1}}
\newcommand{\p}{\pi} 

\newcommand{\vv}[1]{\mathbf{#1}} 

\newcommand{\all}[1]{\p^{#1}, \theta^{#1}, \varphi^{#1}}

\definecolor{brownn}{rgb}{0.66,0.33,0.00}
\usepackage{setspace} 
\usepackage[]{geometry} 

\newcommand\del{\bgroup\markoverwith 
{\textcolor{red}{\rule[2pt]{2pt}{0.5pt}}}\ULon}

\title{A smoothed semiparametric likelihood for estimation of
  nonparametric finite mixture models with a copula-based dependence
  structure}

\author{Michael Levine\thanks{Department of Statistics, Purdue
        University, West Lafayette, IN, 47906}
   \and Gildas Mazo\thanks{Université Paris-Saclay, INRAE, MaIAGE, 78350, Jouy-en-Josas,
        France}}

\begin{document}

\maketitle

\begin{abstract}
  In this manuscript, we consider a finite multivariate nonparametric
  mixture model where the dependence between the marginal densities is
  modeled using the copula device. Pseudo EM stochastic algorithms
  were recently proposed to estimate all of the components of this
  model under a location-scale constraint on the marginals. Here, we
  introduce a deterministic algorithm that seeks to maximize a
  smoothed semiparametric likelihood. No location-scale assumption is
  made about the marginals. The algorithm is monotonic in one special
  case, and, in another, leads to ``approximate monotonicity''---whereby the
  difference between successive values of the objective function
  becomes non-negative up to an additive term that becomes negligible
  after a sufficiently large number of iterations. The
  behavior of this algorithm is illustrated on  several
  simulated datasets. The results suggest that, under suitable
   conditions, the proposed algorithm may indeed be
  monotonic in general. A discussion of the results and some possible
  future research directions round out our presentation.
\end{abstract}


\section{Introduction}

Let 
\begin{align}
  \label{eq:generic-multivariate-mixture-model}
  g(\vv x) = g(x_{1},\dots,x_{d}) = \sum_{k=1}^{K} \p_{k} f_{k}(x_{1},\dots,x_{d})
\end{align}
be a multivariate mixture model with $K$ components (or clusters---we
shall use these two words interchangeably). We view the model \eqref{eq:generic-multivariate-mixture-model} as a
nonparametric mixture model where individual components $f_{k}$ are
not defined as belonging to any specific parametric family. The research on
selecting the number of components for non- and semiparametric density
mixtures is currently at a very early stage; some developments in this
area can be found in e.g. \cite{kasahara2014non} and
\cite{kwon2021estimation}. Due to this, we assume that the number of
components $K$ is fixed and known in our model.  In general, most of the work on nonparametric mixture modeling so far
assumed that the marginal distributions $f_{k1},\ldots,f_{kd}$ of each
component are conditionally independent. Such an assumption implies
that, conditional on knowing which component a particular observation
has been generated from, its distribution is equal to the product of
its marginals. More formally, this means that
\[
g(\vv x)=\sum_{k=1}^{K}\p_{k}\prod_{j=1}^{d}f_{kj}(x_j).
\]
The conditions sufficient to ensure identifiability for the conditionally independent model 
are known \cite{allman2009identifiability}. There are also a number of
approaches to estimating their parameters~\cite{xiang_overview_2019}, both iterative
\cite{benaglia2009like, levine2011maximum} and closed form solutions
\cite{bonhomme2016non}. However, the assumption of conditional
independence is not always a realistic one. For example, it is
unlikely to be true when dealing with RNA-seq data
\cite{rau2015co}. Thus, it seems desirable to relax this assumption
while retaining the generality of the nonparametric approach.

To the best of our knowledge, the only known results on estimation of
nonparametric mixture models with conditionally non-independent
components are
\cite{mazo2017semiparametric,mazo2019constraining}. They consider a
special case of the general nonparametric mixture model, allowing for
a non-trivial dependence structure where the marginals are assumed to
belong to a location-scale family. Stochastic algorithms were proposed
to estimate the copula parameter and the nonparametric marginals. The
estimation algorithms, while performing well in practice, do not
optimize any particular objective function. Because of this, their
convergence analysis will necessarily be a difficult one. In this
manuscript, our goal is to suggest a deterministic algorithm capable
of estimating the components of a nonparametric mixture model with
conditionally non-independent components without a location-scale
assumption for the marginals, since such an assumption is far from
commonly satisfied in applications.

In order to continue, we are going to fix the notation first. It is
well-known that, due to Sklar's theorem \cite{nelsen2007introduction}
p. $18$, every $d-$ dimensional multivariate cumulative distribution
function can be represented as a copula of the corresponding marginal
cumulative distribution functions. Indeed, let
$F_{k1}(x_1),\ldots,F_{kd}(x_d)$ be the marginal cumulative
distribution functions of the cumulative distribution function
$F_k(x_1,\ldots,x_d)$ that corresponds to the density
$f_{k}(x_{1},\dots,x_{d}).$ Then, there exists a $d-$ copula $C_k,$
which is a function $C_k:[0,1]^{d}\rightarrow [0,1],$ such that
\[
F_k(x_1,\ldots,x_d)=C_k(F_{k1}(x_1),\ldots,F_{kd}(x_d)),
\]
see~\cite{nelsen2007introduction} pp. $46.$ If the marginal cumulative
distribution functions are continuous, then the copula is unique.
The 
copula $C_k$ can be viewed as a $d$-dimensional cumulative distribution function
with uniform marginal distributions. Taking the derivative of order
$d$, one immediately obtains the representation
\[
f_{k}(x_{1},\dots,x_{d}) = c_k(F_{k1}(x_{1}),\dots,F_{kd}(x_{d}))
\prod_{j=1}^{d} f_{kj}(x_{j})
\]
where $c_k$ is the density of the copula
$C_k$. We assume that each copula density $c_k$ belongs to some
parametric family of copula densities indexed by a parameter
$\theta_{k}$. Denoting by 
$\varphi$ the set of all marginal densities $\{f_{kj}\}$,
and denoting by $\p=(\p_{1},\dots,\p_{K})^{'}$ and
$\theta=(\theta_{1},\dots,\theta_{K})^{'}$ the vectors of all
 weights and copula parameters, respectively, we have
\begin{align}
  \label{eq:copula-decomposition-for-components}
  f_{k}(\vv x;\theta,\varphi)
  = f_{k}(x_1, \dots, x_d;\theta,\varphi)
  = c(F_{k1}(x_{1}),\dots,F_{kd}(x_{d});\theta_{k})
  \prod_{j=1}^{d} f_{kj}(x_{j}),
\end{align}
so that~\eqref{eq:generic-multivariate-mixture-model}
and~\eqref{eq:copula-decomposition-for-components} define a class of
mixture densities that can be stated as $g(\cdot;\p,\theta,\varphi).$

The rest of this manuscript is structured as follows. Section $2$
introduces a general algorithm that can be used to estimate finite
mixtures of multivariate densities with a dependence structure defined
through the use of copulas. Section $3$ provides some results about
the monotonicity property of two simplified versions of this
algorithm. Section $4$ illustrates the performance of our algorithm
with a simulation study. Section $5$ discusses the results obtained
and suggests possible directions for future research.

\section{Algorithm}\label{sec:algorithm}

The goal of our manuscript is to estimate the components and weights
of the model
\eqref{eq:generic-multivariate-mixture-model}-\eqref{eq:copula-decomposition-for-components}. The
definition of such an algorithm starts with an objective function that
we are going to introduce next. First, let $K(\cdot)$ be a proper
univariate density function that can be used for kernel density
estimation and
$K_{h}(\cdot):=\frac{1}{h}K\left(\frac{\cdot}{h}\right)$ its rescaled
version where $h>0$ is a bandwidth. Next, for a generic function $f,$
we define
\begin{equation}
  \op{N}_hf(x):=\exp\left(\int K_{h}(x-u)\log f(u)\,\diff u\right)
\label{eq:smootheroperator}
\end{equation}
which is a nonlinear smoother of the function $f.$ Note that, even if $f$ is a density, $\op{N}f$ is not, in general a density due to Jensen's inequality. Now, we define the operator $\op{O}$ by 
$\op{O}f_{k}(\vv x;\theta,\varphi) =
c(F_{k1}(x_{1}),\dots,F_{kd}(x_{d});\theta_{k})
\prod_{j=1}^{d}\op{N}f_{kj}(x_{j})$.
This definition allows different bandwidths for different
dimensions and clusters, if needed.
Finally, let us denote
$\check{g}(\vv x; \p, \theta, \varphi) =
\sum_{k=1}^{K}\p_{k}\op{O}f_{k}(\vv x;\theta,\varphi)$.

The objective function we seek to maximize is the population version
of the smoothed semiparametric log-likelihood, given by
\begin{align}
  \ell(\p, \theta, \varphi) =
  \int g(\vv x) \log \frac{\check g(\vv x; \p, \theta, \varphi)}
  {g(\vv x)} \, \diff \vv x,\label{obj_func}
\end{align}
over all $(\p,\theta,\varphi)$; here $g(\mathbf x)$ is the target
density.  If the marginal distributions are conditionally independent
then $c(u_{1},\dots,u_{d};\theta_{k}) \equiv 1$ for every
$\theta_k$ and $k$, and hence~\eqref{obj_func} reduces to the smoothed
semiparametric log-likelihood considered in~\cite{levine2011maximum}.

\begin{lemma}\label{J1}
  For any choice of parameters $\widetilde \p,$ $\widetilde \theta,$
  $\widetilde \varphi,$ the smoothed loglikelihood difference is
  bounded as
  \begin{align*}
    \ell(\p, \theta, \varphi) 
    - \ell(\widetilde \p, \widetilde \theta, \widetilde \varphi)
    \le& \sum_{k=1}^K -\log \frac{\widetilde \p_{k}}{\p_{k}}
         \int g(\vv x) w_k(\vv x; \p, \theta, \varphi) \diff \vv x\\
       &-\int g(\vv x)\sum_{k=1}^{K}w_{k}(\vv x; \p, \theta, \varphi)\log
         \frac{\prod_{j=1}^{d} \op N \widetilde f_{kj}(x_{j})}
         {\prod_{j=1}^{d} \op N f_{kj}(x_{j})}\,\diff \vv x\\
       &-\int g(\vv x)\sum_{k=1}^{K}w_{k}(\vv x; \p, \theta, \varphi)\log \frac{c(\widetilde F_{k1}(x_{1}),
         \dots,
         \widetilde F_{kd}(x_{d}); \widetilde \theta_{k})}
         {c(F_{k1}(x_{1}),\ldots,F_{kd}(x_{d});\theta_{k})}\,\diff\vv x\\
       :=& \Psi_{1}(\widetilde \p | \all{}) + \Psi_{2}(\widetilde \varphi | \all{}) 
           + \Psi_{3}(\widetilde\theta, \widetilde \varphi | \all{}), 
  \end{align*}
  where the distribution functions $\widetilde F_{kj}$ are those
  associated with $\{\widetilde f_{kj}\} = \widetilde \varphi$ and
  \begin{equation}\label{w_def}
    w_{k}(\vv x; \p, \theta, \varphi) = \p_{k} \op O f_{k}(\vv x; \theta,
    \varphi) / \check g(\vv x; \p, \theta, \varphi),
  \end{equation}
  $k=1,\dots,K$.
  \begin{proof}[Proof of Lemma \ref{J1}.] 
    By definition, the difference of smoothed log-likelihoods can be written down as  
    \begin{align*}
      \ell(\p, \theta, \varphi)
      - \ell(\widetilde \p, \widetilde \theta, \widetilde \varphi)
      &= - \int g(\vv x) \log \frac
        {\sum_{k=1}^{K} \widetilde \p_{k} \op O
        f_{k}(\vv x; \widetilde \theta, \widetilde \varphi)}
        {\sum_{k=1}^{K} \p_{k} \op O f_{k}(\vv x; \theta, \varphi)}
        \, \diff \vv x\\
      &= - \int g(\vv x) \log \sum_{k=1}^{K} w_{k}(\vv x; \all{})
        \frac{\widetilde \p_{k} \op O 
        f_{k}(\vv x; \widetilde \theta, \widetilde \varphi)}
        {\p_{k} \op O f_{k}(\vv x; \theta, \varphi)} \, \diff \vv x
    \end{align*}
    At this point, it remains only to apply Jensen's inequality to a
    convex combination on the right-hand side whose coefficients are
    $w_{k}(\vv x; \theta, \varphi).$
  \end{proof}
\end{lemma}

Instead of minimizing
$\ell(\p, \theta, \varphi) - \ell(\widetilde \p, \widetilde \theta,
\widetilde \varphi)$ with respect to
$(\widetilde \p, \widetilde \theta, \widetilde \varphi)$ directly, we
seek to minimize the upper bound proposed by Lemma \ref{J1}. This
approach is in the spirit of MM (Minimization-Majorization)
algorithms; see e.g. \cite{wu2010mm} for the detailed discussion. To
do this, our heuristic is to minimize each of the three terms
$\Psi_{1}(\widetilde \p | \all{})$, $\Psi_{2}(\widetilde \varphi | \all{})$,
$\Psi_{3}(\widetilde\theta, \widetilde \varphi | \all{})$ separately. This is
sometimes called ``minimization by part''. 
To minimize the first term $\Psi_{1}(\widetilde \p | \all{})$, we have to
choose $\widehat \p = \widehat\p$ where
$\widehat\p_{k} = \int g(\vv x) w_{k}(\vv x; \all{}) \, \diff \vv x,$
$k=1,\ldots,K.$ This is the result that can be obtained using standard
constrained optimization techniques. Note that the resulting minimum
must be non-positive since the first term can be made zero by choosing
$\widetilde\p=\p.$ To minimize the second term
$\Psi_{2}(\widetilde \varphi | \all{})$, define, as a first step,
$$\widehat f_{kj}(u_{j}) =\alpha_{kj} \int g(\vv x) w_{k}(\vv x; \all{})
K_{h_{kj}}\left(x_{j} - u_{j}\right) \, \diff \vv x,$$ for any
$k=1,\ldots,K$ and $j=1,\ldots,d$, where $\alpha_{kj}$ is the
normalizing constant ensuring that the newly defined $\widehat f_{kj}$
is, indeed, a proper density function.  Then, we have
\begin{align*}
  &- \int g(\vv x) w_{k}(\vv x; \all{}) \log
  \op N \widetilde f_{kj}(x_{j}) \, \diff \vv x \\
  &= - \int g(\vv x) w_{k}(\vv x; \all{}) 
    \left(
    \int K_{h_{kj}}\left(x_{j} - u_{j}\right)
    \log \widetilde f_{kj}(u_{j}) \, \diff u_{j} \right)
    \, \diff \vv x \\
  &= - \int \log \widetilde f_{kj}(u_{j}) \widehat f_{kj}(u_{j})
    \, \diff u_{j}.
\end{align*}
The same argument as in \cite{levine2011maximum} applies: the quantity above
is minimized if we select $\widetilde f_{kj}(u)=\widehat f_{kj}(u).$ The resulting minimum will also be less than or equal to zero because $\Psi_{2}(\widetilde \varphi | \all{})=0$ when $\widetilde \varphi = \varphi.$

Now, we can propose the following general algorithm for estimation of $(\p,\theta,\varphi).$

\renewcommand\labelenumi{A\theenumi}

\begin{enumerate}
\item\label{step1} Choose initial values $\p^{0},$ $\varphi^{0},$ $\theta^{0}$
\item\label{step2} Compute the initial set of weights
\[
w_{k}(\vv x; \p^{0}, \theta^{0}, \varphi^{0}) = \p_{k}^{0} \op O f_{k}(\vv x; \theta^{0},
\varphi^{0}) / \check g(\vv x; \pi^{0}, \theta^{0}, \varphi^{0}).
\]
\item\label{step3} At any step of iteration $t=1,2,\ldots$ select $$\p_{k}^{t} = \int g(\vv x) w_{k}(\vv x;\p^{t-1},\theta^{t-1},\varphi^{t-1}) \, \diff \vv x,$$ $k=1,\ldots,K.$ 
\item Select as the next value of the density function vector $\varphi^{t}=\{f_{kj}^{t}\}$ where 
$$f_{kj}^{t}(u_{j}) = \alpha_{kj}\int g(\vv x) w_{k}(\vv x;\p^{t-1},\theta^{t-1},\varphi^{t-1})
  K_{h_{kj}}\left(x_{j} - u_{j}\right)
   \, \diff \vv x$$
   where $\alpha_{kj}$ is the normalizing constant ensuring that the newly defined function is, indeed, a density function.  As a part of this step, also compute updated cumulative distribution functions $F_{kj}^{t}(u_{j})=\int_{-\infty}^{u_{j}}f_{kj}^{t}(y)\,\diff y.$ 
\item\label{step4} Choose the value 
\[
\theta^{t}=\arg\min_{\theta}\Psi_{3}(\theta, \varphi^{t} | \all{t-1}).
\]
\item\label{step5} Redefine weights
\[
w_{k}(\vv x; \p^{t}, \theta^{t}, \varphi^{t}) = \p_{k}^{t} \op O f_{k}(\vv x; \theta^{t},\varphi^{t}) / \check g(\vv x; \pi^t, \theta^{t}, \varphi^{t}).
\]
and return to step~A\ref{step3}. 
\end{enumerate}

At each step of the algorithm defined above, the marginals are updated
first and independently of the copula parameter. This strategy was
used in~\cite{mazo2017semiparametric,mazo2019constraining}.

\begin{remark}
  In practice, one implements  the empirical version of the algorithm. Every integral of the form
  $\int g(\vv x) \zeta(\vv x)\, \diff \vv x$, where $\zeta$ is some
  arbitrary function, is replaced by
  $\tfrac{1}{n} \sum_{i=1}^n \zeta(\vv X_i)$, where $\vv X_i =
  (X_{i1}, \dots, X_{id})$, $i = 1, \dots, n$, are observations from
  the target density $g$. The objective function to be maximized is
  then the empirical version of the smoothed log-likelihood, given by
  $\tfrac{1}{n} \sum_{i=1}^n \log \check g(\vv X_i; \p, \theta,
  \varphi)$ (up to an additive constant). Here the bandwidths of the nonlinear
  smoothers are allowed to depend on the data.
\end{remark}

\section{Studying the algorithm}

Whether the algorithm proposed in  Section~\ref{sec:algorithm} is monotonic  with respect to the objective functional \eqref{obj_func} is an open question. In some special cases, the answer is positive. One such case that we identified is when probabilities $\pi_{k}$ and the marginal densities $f_{kj}$ are known beforehand. In such a case, the simplified algorithm is as follows. 

\renewcommand\labelenumi{B\theenumi}

\begin{enumerate}
\item\label{step11} Choose initial value of the copula parameter $\theta^{0}$.
\item\label{step22} Compute the initial set of weights
\[
w_{k}(\vv x; \p, \theta^{0}, \varphi) = \p_{k} \op O f_{k}(\vv x; \theta^{0},
\varphi) / \check g(\vv x; \p, \theta^{0}, \varphi).
\]
\item\label{step41} For any $t=1,2,\ldots$ choose the value 
\[
\theta^{t}=\arg\min_{\theta}\Psi_{3}(\theta, \varphi | \p, \theta^{t-1}, \varphi).
\]
\item\label{step51} Redefine weights
\[
w_{k}(\vv x; \p, \theta^{t}, \varphi) = \p_{k} \op O f_{k}(\vv x; \theta^{t},\varphi) / \check g(\vv x; \p, \theta^t, \varphi).
\]
and return to step~B\ref{step41}.
\end{enumerate} 
\begin{proposition}
The algorithm defined in~B\ref{step11}--B\ref{step51} is monotonic
with respect to $\theta$,
that is, $\ell(\p, \theta^{t-1}, \varphi) - \ell(\p, \theta^t,
\varphi) \le 0$ for every $t=1, 2, \dots$
\end{proposition}
\begin{proof}
The smoothed likelihood difference is bounded from above as  
\begin{align*}
 \ell(\p, \theta, \varphi)- \ell(\p, \widetilde \theta, \varphi)
 &\le \Psi_{3}(\widetilde\theta,\varphi | \all{}) \\
  &=-\int g(\vv x)\sum_{k=1}^{K}w_{k}(\vv x; \all{})\log \frac{c(F_{k1}(x_{1}),
    \dots,
    F_{kd}(x_{d}); \widetilde \theta_{k})}
    {c(F_{k1}(x_{1}),\ldots,F_{kd}(x_{d});\theta_k)}\,\diff\vv x.
\end{align*}
Choosing
$\theta^{*}=\arg\min_{\widetilde\theta}\Psi_{3}(\widetilde\theta,\varphi
| \all{})$ produces
$$\ell(\p, \theta, \varphi) - \ell(\p, \theta^{*},\varphi) \le
\Psi_{3}(\theta^{*},\varphi | \all{}) =
\min_{\widetilde\theta}\Psi_{3}(\widetilde\theta,\varphi | \all{});$$ since
there exists a value $\widetilde\theta=\theta$ such that
$\Psi_{3}(\theta,\varphi | \all{})\equiv 0$, the minimal value of
$\Psi_{3}(\widetilde\theta,\varphi | \all{})$ will be less than or equal to
zero.
\end{proof}

Another interesting special case results when one assumes that both component weights $\pi_{k}$ and  copula parameters $\theta_{k}$ are known while the marginal densities $f_{kj}$ are unknown. In this case, the simplified algorithm will be as follows. 

\renewcommand\labelenumi{C\theenumi}

\begin{enumerate}
\item\label{step31} Choose initial values $\varphi^{0}$
\item\label{step32} Compute the initial set of weights
\[
w_{k}(\vv x; \p, \theta, \varphi^{0}) = \p_{k}\op O f_{k}(\vv x; \theta,
\varphi^{0}) / \check g(\vv x; \p, \theta, \varphi^{0}).
\]
\item\label{step33} For $t=1,2,\ldots$ select as the next value of the
  density function vector $\varphi^{t}=\{f_{kj}^{t}\}$ where
  $f_{kj}^{t}(u_{j}) = \alpha_{kj}\int g(\vv x) w_{k}(\vv x; \p,
  \theta, \varphi^{t-1})
  K_{h_{kj}}\left(x_{j} - u_{j}\right)
   \, \diff \vv x$. Here, $\alpha_{kj}$ is a normalizing constant, ensuring that the newly defined function is, indeed, a density function.  As a part of this step, also compute updated cumulative distribution functions $F_{kj}^{t}(u_{j})=\int_{-\infty}^{u_{j}}f_{kj}^{t}(y)\,\diff y.$ 
\item\label{step35} Redefine weights
\[
w_{k}(\vv x; \p, \theta, \varphi^{t}) = \p_{k}\op O f_{k}(\vv x; \theta,\varphi^{t}) / \check g(\vv x; \p, \theta, \varphi^{t}).
\]
and return to step C\ref{step33}. 
\end{enumerate}
The special case of the general algorithm defined above possesses an ``approximate monotonicity'' property in the following sense. 

\begin{proposition}\label{func_version}
We assume that the target density $g(\vv x)$ has a compact support
$\Omega$. We also assume that none of the known weights $\pi_{k}$ is
equal to zero.
Suppose that the kernel function $K(\cdot)$
is a proper density function defined on $[-1,1]$, 
bounded away from zero by $K_{*}>0$, and 
Lipschitz continuous with a positive Lipschitz constant $L$.
We assume that the copula density function
$c(u_1,\ldots,u_d;\theta)$ is also Lipschitz continuous on $[0,1]^{d}$
and bounded away from zero.  Then, there exists a subsequence $\varphi^{t_l}
= (f_{kj}^{t_{l}},\, k = 1, \dots, K,\, j = 1, \dots, d)$, $l = 1, 2, \dots$, such that the the algorithm C\ref{step31}--C\ref{step35} is ``approximately monotonically ascending'' along this subsequence: 
\[
  \ell(\p, \theta, \varphi^{t_{l-1}}) - \ell(\p, \theta, \varphi^{t_{l}}) \le o(1)
\]
as $l\rightarrow \infty$.
\end{proposition}

\begin{remark}
  It follows directly from the definition that
  $K_{*}\le K(\cdot)\le K^{*}$ where both $K_{*}$ and $K^{*}$ are
  positive. The assumptions of Lipschitz continuity and boundedness
  away from zero for the kernel function $K(\cdot)$ do not represent a practical problem since they are
  not concerned with the actual data---rather, $K(\cdot)$ is a tool
  used to analyze the data. Our simulation results suggest that they also may not be necessary.
\end{remark}

\begin{remark}
  The assumption of compact support for the target density $g(\vv x)$
  and, by extension, for all of the marginal densities
  $f_{kj}$ does not represent a problem from the
    practical viewpoint. From the theoretical viewpoint,
  a result analogous to Proposition~\ref{func_version}
  can be proved if one assumes that all of the marginal densities
  decay to zero sufficiently fast at infinity and using the
  Fr\'{e}chet-Kolmogorov theorem instead of the Arzel\`{a}-Ascoli
  theorem \cite{brezis2011functional} p. $126$.
\end{remark}

\begin{remark}
As an example of copulas satisfying conditions of Proposition~\ref{func_version} we can point out Farlie-Gumbel-Morgenstern (FGM) copulas as well as so-called copulas with cubic sections (that are direct generalizations of FGM copulas) \cite{nelsen2007introduction} pp. $77-84$.
\end{remark}
\begin{proof}
  The difference in log-likelihoods can be bounded as 
  \begin{align*}
 & \ell(\p, \theta, \varphi^{t_{l-1}}) - \ell(\p, \theta, \varphi^{t_l}) 
   \le \Psi_{2}(\varphi^{t_l} | \p, \theta, \varphi^{t_{l-1}})
          +\Psi_{3}(\theta, \varphi^{t_l} | \p, \theta, \varphi^{t_{l-1}}) \\
    &=-\int g(\vv x)\sum_{k=1}^{K}w_{k}(\vv x; \p, \theta, \varphi^{t_{l-1}})\log
        \frac{\prod_{j=1}^{d} \op N f_{kj}^{t_l}(x_{j})}
        {\prod_{j=1}^{d} \op N f_{kj}^{t_{l-1}}(x_{j})}\,\diff\vv x \\
    & -\int g(\vv x)\sum_{k=1}^{K}w_{k}(\vv x; \p, \theta, \varphi^{t_{l-1}})\log
        \frac{c(F_{k1}^{t_l}(x_{1}),\dots,F_{kd}^{t_l}(x_{d}); \theta_{k})} 
        {c(F_{k1}^{t_{l-1}}(x_{1}),\ldots,F_{kd}^{t_{l-1}}(x_{d});\theta_k)}\,\diff\vv x. 
  \end{align*}
  Recall that minimization of
  $\Psi_{2}(\varphi^{t_l} | \p, \theta, \varphi^{t_{l-1}})$ always results in
  $\Psi_{2}(\varphi^{t_l} | \p, \theta, \varphi^{t_{l-1}})\le 0$ since the
  choice $f_{kj}^{t_l}=f_{kj}^{t_{l-1}}$ for all $k=1,\ldots,K$ and
  $j=1,\ldots,d$ makes this term equal to zero. Therefore, it remains to
  show that
  $\Psi_{3}(\theta, \varphi^{t_l} | \p, \theta, \varphi^{t_{l-1}}) \to
  0$ as $l \to \infty$. To do this, let us introduce a lemma.

  \begin{lemma}
    \label{lem:univergseq}
    For each $k = 1, \dots, K$ and $j = 1, \dots, d$, the sequence
    $f_{kj}^t$, $t=1, 2, \dots$ has a uniformly converging subsequence
    $f_{kj}^{t_l}$, $l = 1, 2, \dots$.
  \end{lemma}

  The proof of Lemma~\ref{lem:univergseq} is similar
  to the proof of Lemma A2 in ~\cite{levine2011maximum} and is not given. Denote by $f_{kj}^*$
  the limit of $f_{kj}^{t_l}$ as $l \to \infty$. Denote by $\varphi^*$
  the collection of all such limits. Since $\Omega$ is
  compact, it follows in a straightforward manner from Lemma~\ref{lem:univergseq}
  that each subsequence $F_{kj}^{t_{l}}(u)$ converges uniformly to
  $F_{kj}^{*}(u) := \int_{-\infty}^u f_{kj}^*(x)\, \diff x$.
  To show that $\Psi_{3}(\theta, \varphi^{t_l} | \p, \theta,
  \varphi^{t_{l-1}})$ goes to zero as $l$ goes to infinity, we proceed
  as follows. We have
  \newcommand\cop[1]{c(F_{k1}^{#1}(x_1), \dots,
    F_{kd}^{#1}(x_d) ; \theta_k)}
  \newcommand\wei[1]{w_k(\vv x; \p, \theta, \varphi^{#1})}
  \begin{align*}
    &| \Psi_3(\theta, \varphi^{t_l} | \p, \theta, \varphi^{t_{l-1}}) | \\
   & \le \sum_{k=1}^K \left |\int g(\vv x) \wei{t_{l-1}}
        \log \frac{\cop{t_l}}{\cop{t_{l-1}}}\, \diff \vv x \right |.
  \end{align*} 
  Each summand is bounded as
  \begin{align}
    \label{eq:bound-with-two-terms}
    &\left| \int g(\vv x) \wei{*} \log \frac{ \cop{t_l} }{
       \cop{t_{l-1}} } \, \diff \vv x \right| + \nonumber\\
    &\left| \int g(\vv x) (\wei{t_{l-1}} - \wei{t_l})
       \log \frac{ \cop{t_l} }{ \cop{t_{l-1}} } \, \diff \vv x \right|.
  \end{align}
  Since the copula density is bounded from above and below, the second
  term is less than or equal to a constant
  times
     $\int g(\vv x) |\wei{t_{l-1}} - \wei{t_l}|\, \diff \vv x$.
  But, by the dominated convergence theorem, this integral vanishes
  because the kernel $K$ and the copula density are bounded from above
  and below, the copula density is Lipschitz continuous and,
  from~\cite{levine2011maximum}, $\op N f_{kj}^{t_l}$ converges
  uniformly to $\op N f_{kj}^*$ as $l \to \infty$. 

  The first term in~(\ref{eq:bound-with-two-terms}) is bounded by
  \begin{align*}
   & \left| \int g(\vv x) \wei{*} \log \frac{ \cop{t_l} }{ \cop{*} }
      \, \diff \vv x \right| \\
    &+ \left| \int g(\vv x) \wei{*} \log \frac{ \cop{*} }{ \cop{t_{l-1}} }
      \, \diff \vv x \right|.
  \end{align*}
  But again this bound goes to zero by similar arguments.  This
  finishes the proof.
\end{proof}

\section{Numerical Study}\label{Num}

Five hundred replications of four independent artificial datasets of
sizes $n=300, 500, 700, 900$ were generated from the mixture
model~(\ref{eq:generic-multivariate-mixture-model})--(\ref{eq:copula-decomposition-for-components})
with 3 clusters of equal proportions, FGM copulas with parameters
$-0.5, 0.5, 0$ and marginals as in Table~\ref{tab:marginals}, where
$\mathrm{N}(\mu,\sigma^2)$ and $\mathrm{L}(\mu,\sigma^2)$ refer to the
normal and Laplace distributions with mean $\mu$ and standard
deviation $\sigma$, respectively. (The density of a $\mathrm{L}(\mu,\sigma^2)$
 distribution is then given by
 $f(x)=e^{-\sqrt 2 |x - \mu| / \sigma} / (\sqrt 2 \sigma)$ for any real $x$.)
The
algorithm of Section~\ref{sec:algorithm} was implemented to estimate
the cluster proportions, the copula parameters and the marginal
densities. The kernel $K$ was the Gaussian kernel. 
The number of iterations was arbitrarily fixed to fifty. A
bottleneck of the algorithm is the numerical evaluation of the
integral~\eqref{eq:smootheroperator}. It was found empirically that,
instead of~\eqref{eq:smootheroperator}, evaluating the integral
$\int_{-1.96h}^{1.96h} K_h(u) \log \max\{ f(x - u), 10^{-5} \}\, \diff
u$ gave more stable results more rapidly. The \texttt{integrate}
function of \texttt{R} with the default parameters was used.

 For initialization, a $k$-means algorithm was performed. The
marginal densities were initialized by standard kernel density
estimation using the split returned by the $k$-means algorithm.
For each cluster and dimension, a bandwidth was selected and standard kernel density estimation
performed using only the data assigned to the given
cluster. The bandwidths were kept fixed throughout the algorithm.
The copula parameters were initialized to zero. The cluster
proportions were initialized to the cluster proportions found by the
$k$-means algorithm. 

\begin{table}[htp]
  \centering
  \begin{tabular}{l|ccc}
    &cluster 1           &cluster 2           &cluster 3\\ \hline
    dim 1&$\mathrm{N}(-3,2^2)$&$\mathrm{N}(0,0.7^2)$&$\mathrm{N}(3,1.4^2)$\\
    dim
    2&$\mathrm{L}(0,0.7^2)$&$\mathrm{L}(3,1.4^2)$&$\mathrm{L}(0,2.8^2)$\\
  \end{tabular}
  \caption{Marginals used for the numerical experiment.}
  \label{tab:marginals}
\end{table}

Figure~\ref{fig:smoothed-log-likelihood} shows the values of the
empirical smoothed log-likelihood~(\ref{obj_func}) at each step of the
algorithm for the first ten replications in the case $n=300$ and
$n=900$. All of the trajectories look monotonic.  It was numerically
calculated that, out of the $N=500$ trajectories, only 17 were
non-monotonic for $n=300$ at the $10^{-5}$ precision. This number goes
down to 1 for $n=500$, and zero for $n=700$ and $n=900$. This suggests
that the algorithm of Section~\ref{sec:algorithm} may indeed be
monotonic for the copula and marginal families chosen above.

\begin{figure}[h]
  \centering
  \includegraphics[width=0.55\textwidth]{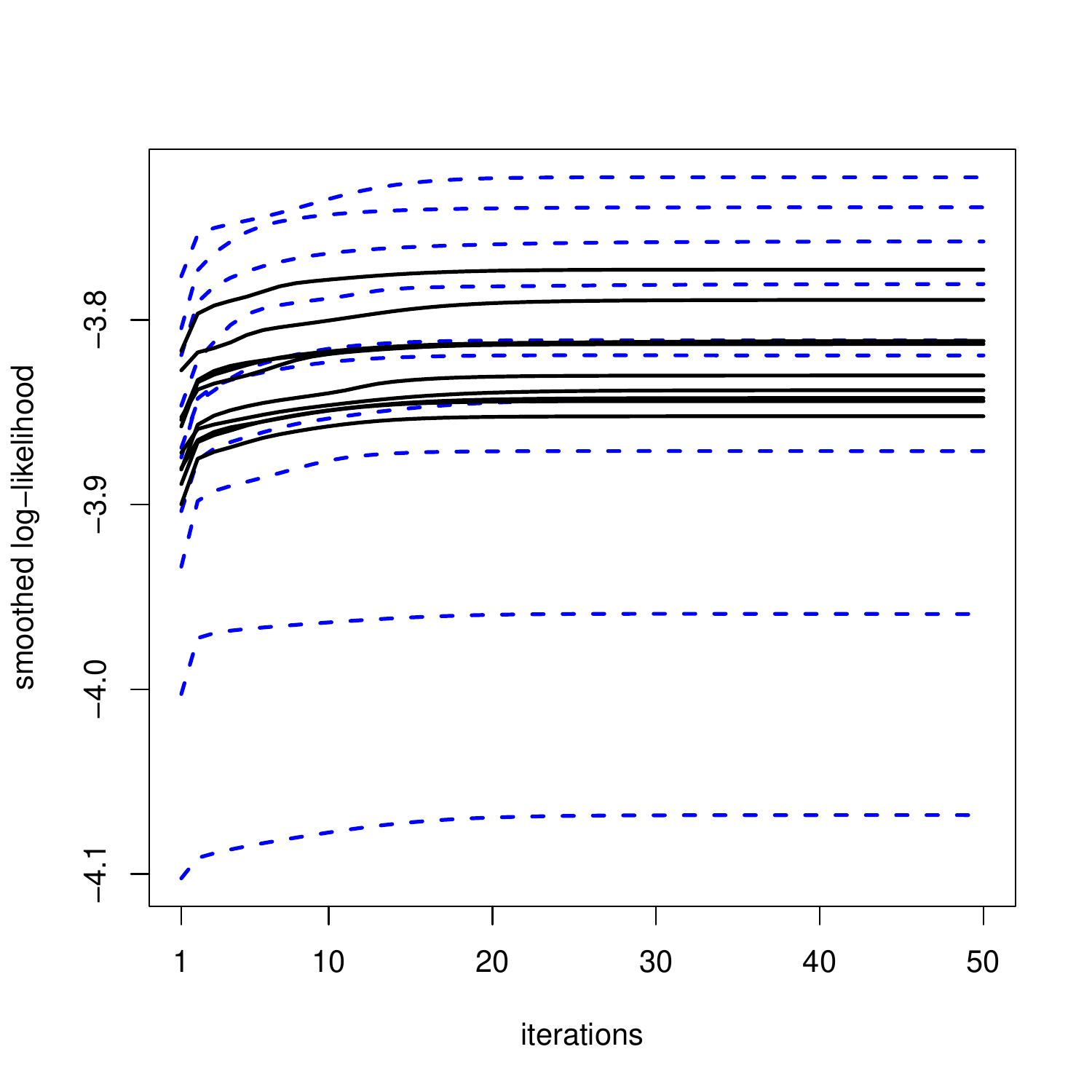}
  \caption{Values of the empirical smoothed log-likelihood at each step of the
    algorithm, for the first ten replications. Black plain lines:
    $n=900$. Blue dotted lines: $n=300$.}
  \label{fig:smoothed-log-likelihood}
\end{figure}

Figure~\ref{fig:coppar} shows the sum of the estimated squared biases
and variances for the copula parameter vector. The variance is
$3$ times higher than the squared bias for $n=300$, and only
$1.6$ times higher for $n=900$. While the bias remains stable, the
variance decreases with $n$, but at a slower rate than the
``parametric'' rate $1/n$. While $n=900$ is $3$ times larger than
$n=300$, the variance at $n=900$ is only $1.5$ times smaller than the
variance at $n=300$. The mean absolute bias is about
$\sqrt{0.5/3}\approx 0.4$, while the mean standard errors at $n=300$
and $n=900$ are about $\sqrt{1.2/3}\approx 0.63$ and
$\sqrt{0.8/3}\approx 0.52$, respectively.

\begin{figure}[h]
  \centering
  \includegraphics[width=0.55\textwidth]{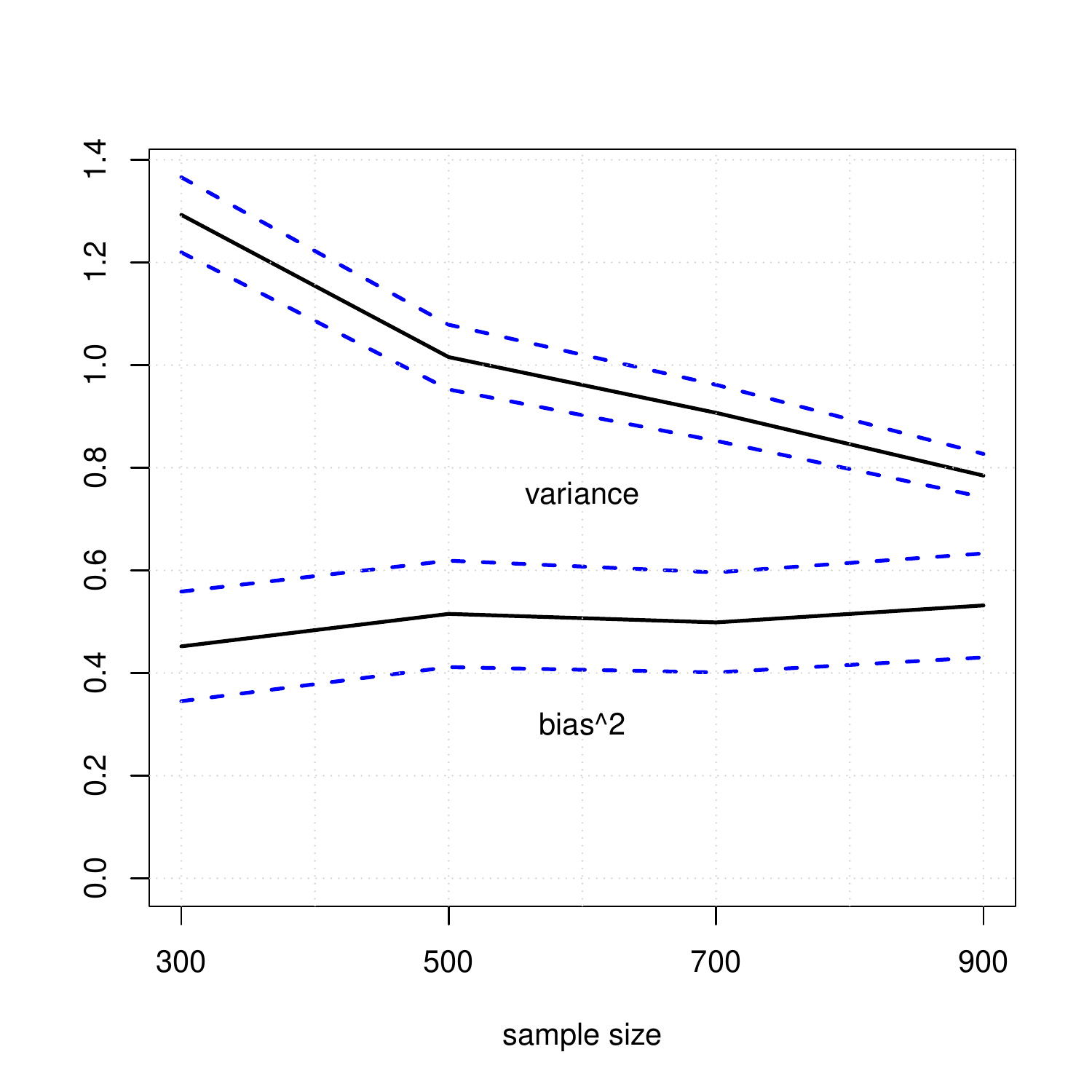}
  \caption{Estimated squared bias and variance of the copula parameter
    vector estimator for various sample sizes at the
    last step of the algorithm. Dashed blue lines represent 95\%
    confidence bands (aka simultaneous confidence intervals) obtained
    from an application of the multivariate central limit theorem to
    the five hundred replications.}
  \label{fig:coppar}
\end{figure}

Figure~\ref{fig:marginals-lstep-n900} shows the marginal
density estimates at the last step of the algorithm for $n=900$, for the
last replication. The estimates agree well with the true marginal
densities. We noticed, however,  that
they were similar to the initial estimates.

\begin{figure}[h]
  \centering
  \subfloat[]{
    \label{fig:subfig:truemargfirstdim}
    \includegraphics[width=.5\textwidth]{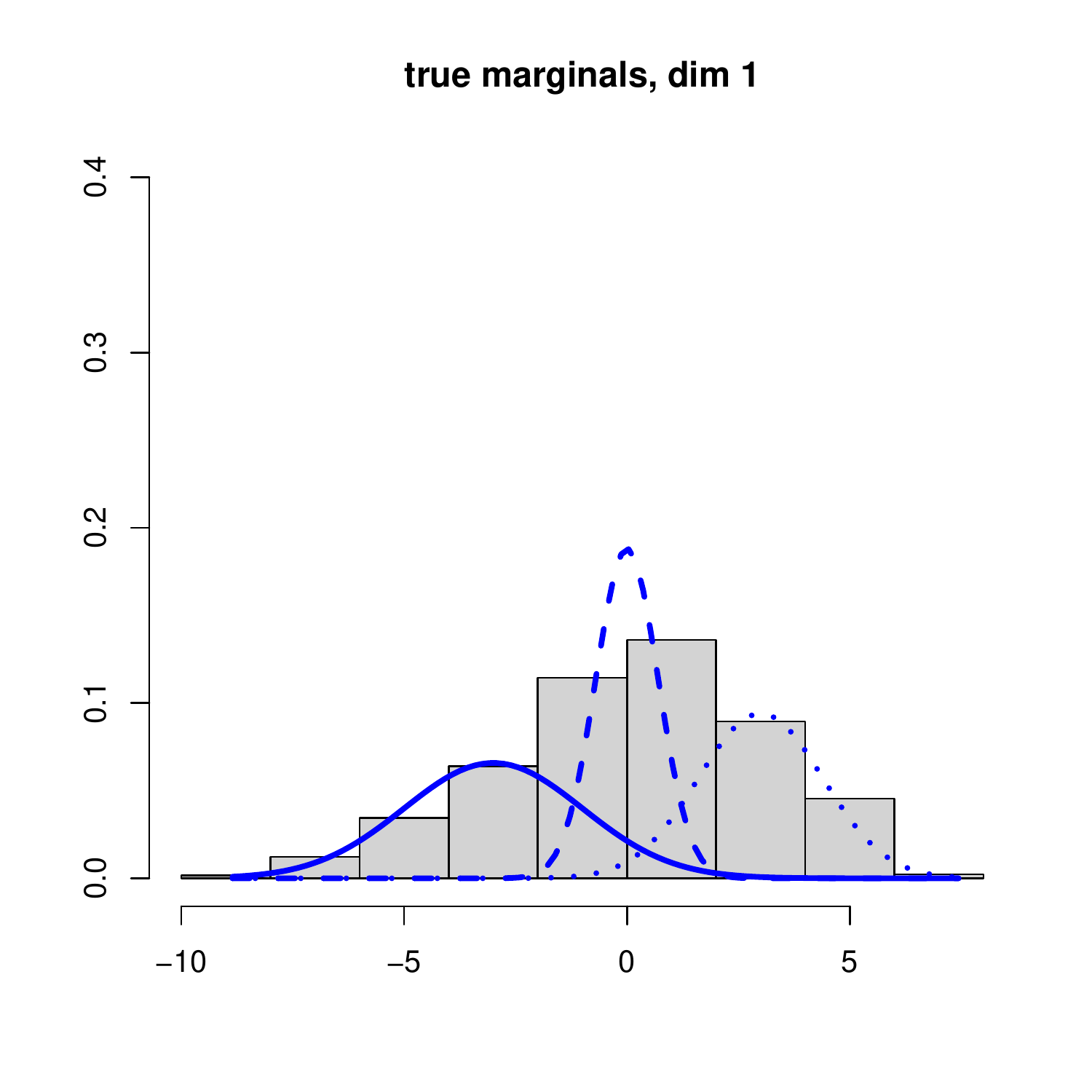} }
  \subfloat[]{
    \label{fig:subfig:truemargseconddim}
    \includegraphics[width=.5\textwidth]{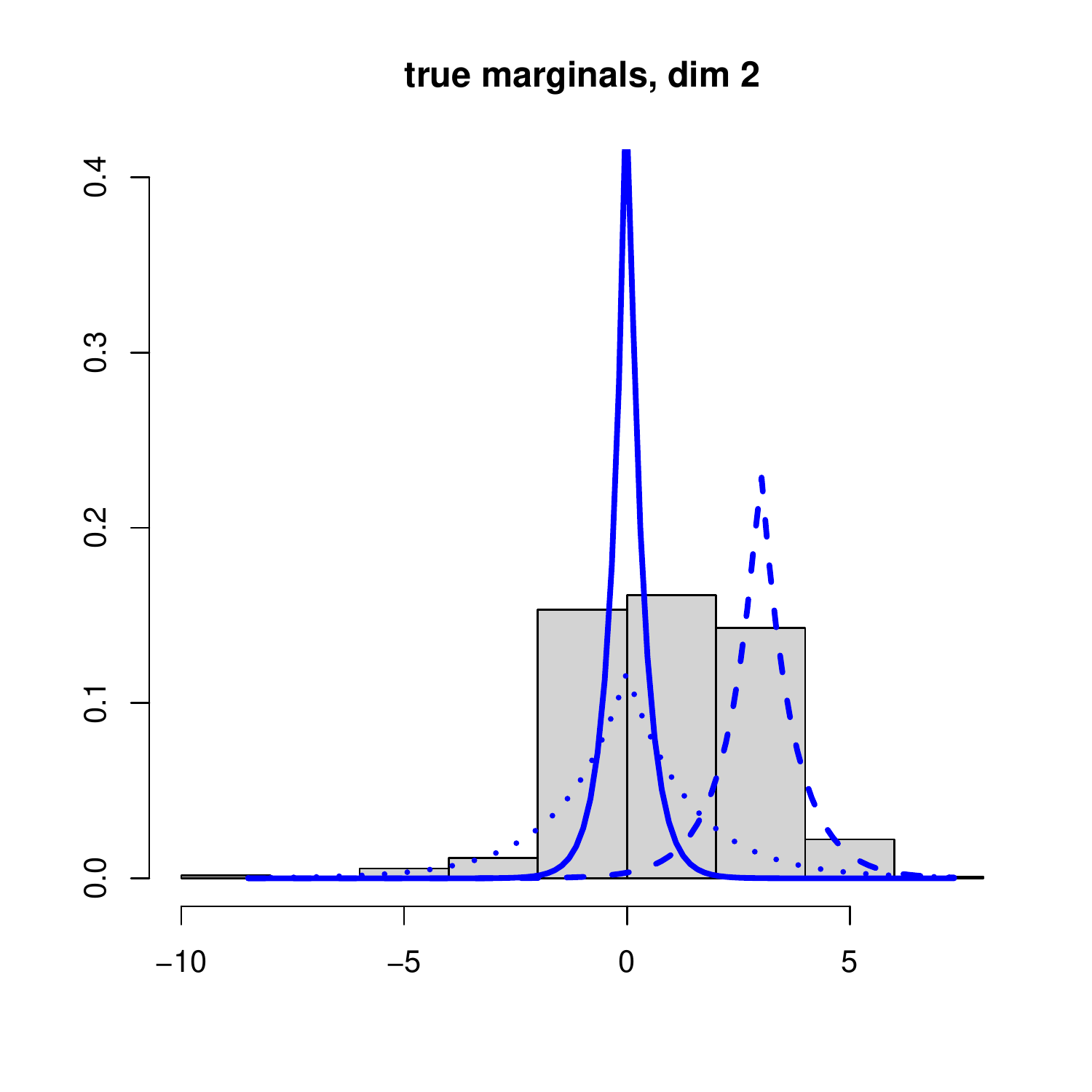} }\\
  \subfloat[]{
    \label{fig:subfig:estimmargfirstdim}
    \includegraphics[width=.5\textwidth]{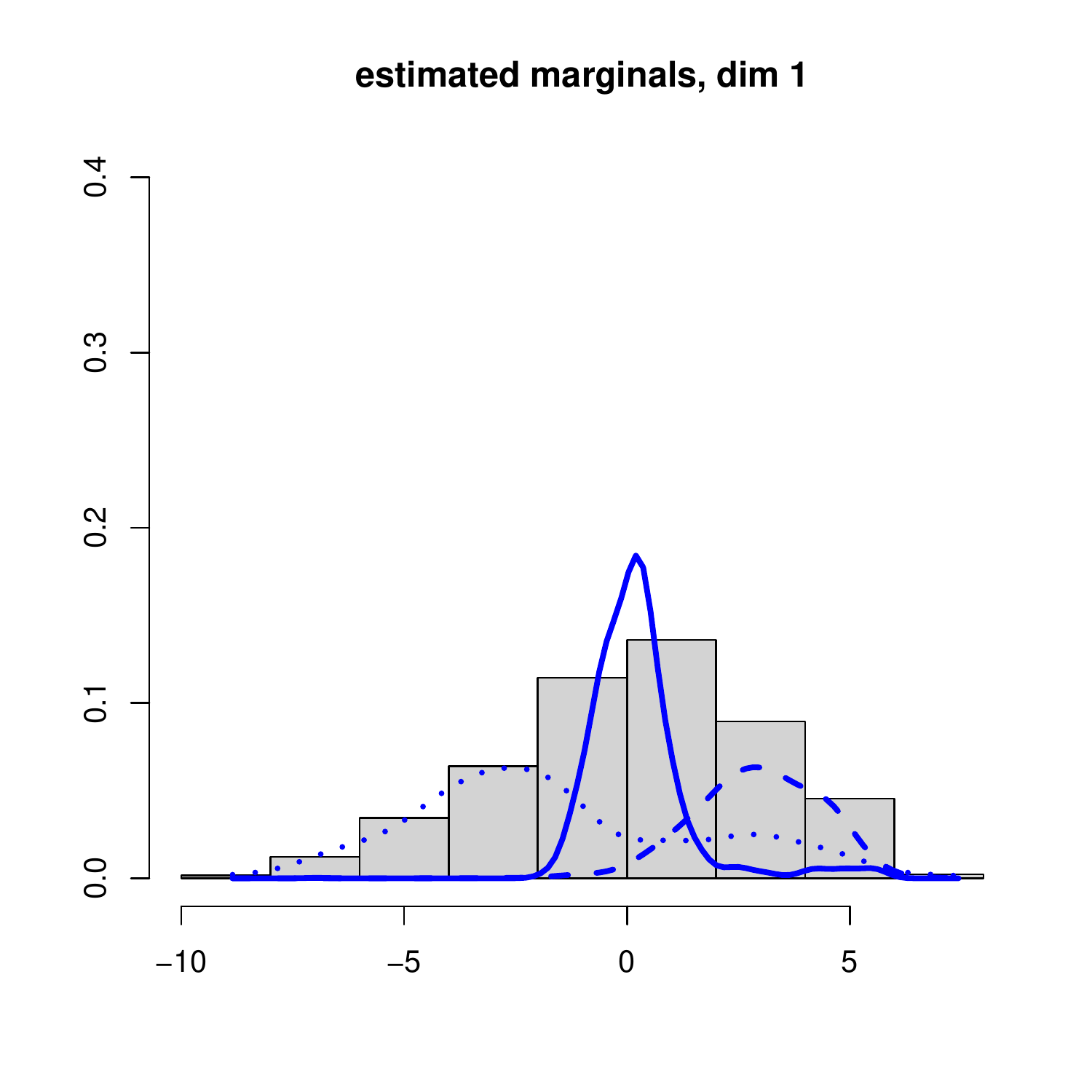} }
  \subfloat[]{
    \label{fig:subfig:estimmargseconddim}
    \includegraphics[width=.5\textwidth]{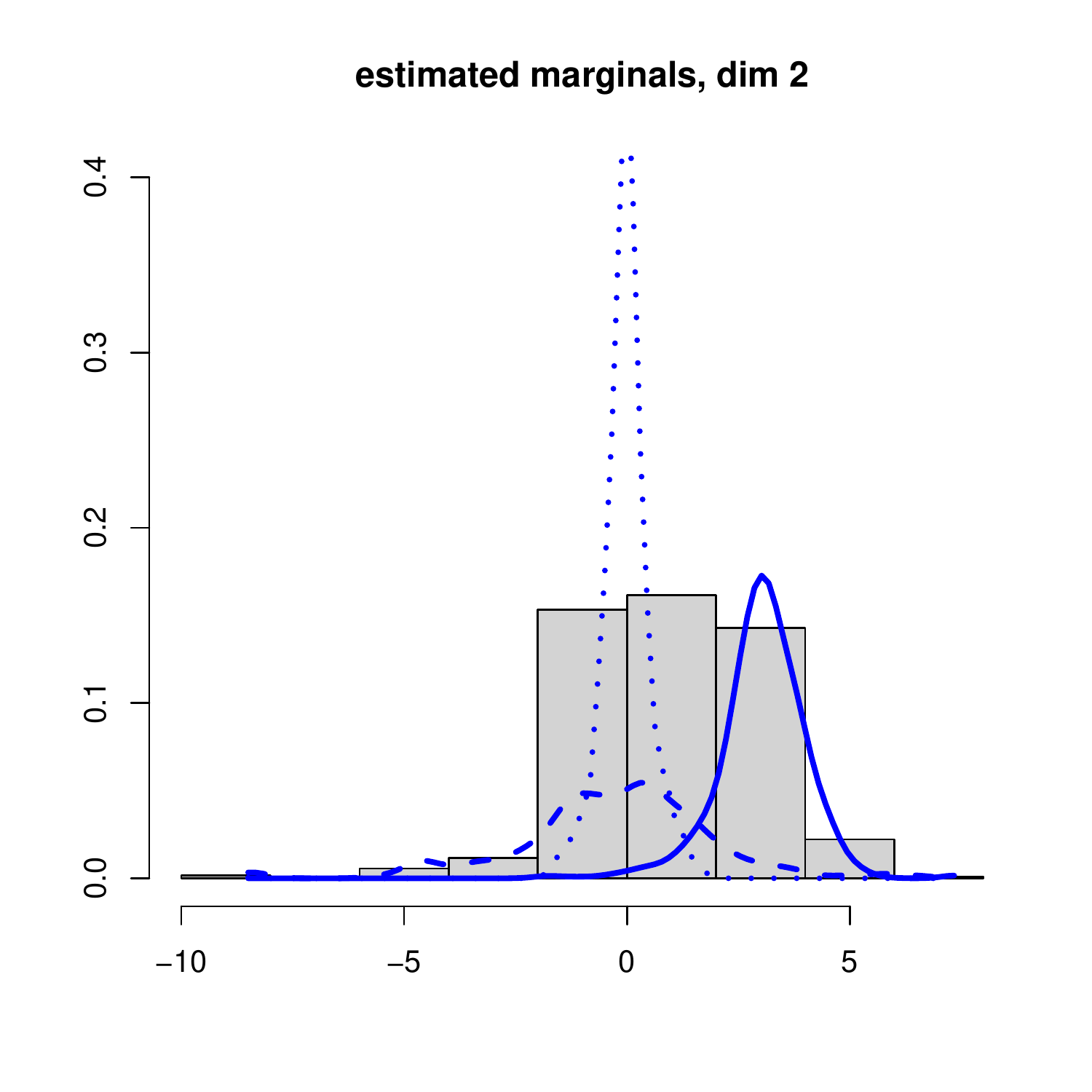} }
  \caption{True and estimated marginal densities of the three clusters
    and the two dimensions for $n=900$ (last replication). The
    top row contains the true marginals and the column on the
    left contains the first dimension. The marginal estimates are
    those found at the last step of the algorithm.}
  \label{fig:marginals-lstep-n900}
\end{figure}


\section{Conclusion}

An algorithm was designed and implemented to estimate the parameters
of copula-based semiparametric mixture models. The model considered is a very general one since it does not assume any specific structure (such as the location-scale assumption) on marginal densities. The algorithm is deterministic, and hence always
returns the same result if fed with the same initial point. Good
performance was obtained in an illustrative numerical example, which
suggests that the algorithm may indeed be monotonic under
appropriate conditions.

However, its theoretical analysis proved to be challenging and only
partial results were obtained for versions of the algorithm where
either the copula parameter or the marginals were fixed. A future
avenue of research may consist of rejecting those updates where the
smoothed log-likelihood does not increase and investigate whether
convergence results of~\cite{meyer1976sufficient,zangwill1969book}
could be applied. To simplify, the full parametric case may first
be considered. To improve the numerical implementation of the
algorithm, the integral~\eqref{eq:smootheroperator} may be computed
with other methods, such as~\cite{qiang_high-order_2010}.
The sensitivity of the algorithm with respect to initialization may be
investigated further, however.

\clearpage

\bibliography{reference}

\begin{thebibliography}{}

\bibitem[\protect\astroncite{Allman et~al.}{2009}]{allman2009identifiability}
Allman, E.~S., Matias, C., and Rhodes, J.~A. (2009).
\newblock Identifiability of parameters in latent structure models with many
  observed variables.
\newblock {\em The Annals of Statistics}, 37(6A):3099--3132.

\bibitem[\protect\astroncite{Benaglia et~al.}{2009}]{benaglia2009like}
Benaglia, T., Chauveau, D., and Hunter, D.~R. (2009).
\newblock An {EM}-like algorithm for semi-and nonparametric estimation in
  multivariate mixtures.
\newblock {\em Journal of Computational and Graphical Statistics},
  18(2):505--526.

\bibitem[\protect\astroncite{Bonhomme et~al.}{2016}]{bonhomme2016non}
Bonhomme, S., Jochmans, K., and Robin, J.-M. (2016).
\newblock Non-parametric estimation of finite mixtures from repeated
  measurements.
\newblock {\em Journal of the Royal Statistical Society: Series B (Statistical
  Methodology)}, 78(1):211--229.

\bibitem[\protect\astroncite{Brezis}{2011}]{brezis2011functional}
Brezis, H. (2011).
\newblock {\em Functional analysis, {S}obolev spaces and partial differential
  equations}.
\newblock Springer.

\bibitem[\protect\astroncite{Kasahara and Shimotsu}{2014}]{kasahara2014non}
Kasahara, H. and Shimotsu, K. (2014).
\newblock Non-parametric identification and estimation of the number of
  components in multivariate mixtures.
\newblock {\em Journal of the Royal Statistical Society: Series B (Statistical
  Methodology)}, 76(1):97--111.

\bibitem[\protect\astroncite{Kwon and Mbakop}{2021}]{kwon2021estimation}
Kwon, C. and Mbakop, E. (2021).
\newblock Estimation of the number of components of nonparametric multivariate
  finite mixture models.
\newblock {\em The Annals of Statistics}, 49(4):2178--2205.

\bibitem[\protect\astroncite{Levine et~al.}{2011}]{levine2011maximum}
Levine, M., Hunter, D.~R., and Chauveau, D. (2011).
\newblock Maximum smoothed likelihood for multivariate mixtures.
\newblock {\em Biometrika}, 98(2):403--416.

\bibitem[\protect\astroncite{Mazo}{2017}]{mazo2017semiparametric}
Mazo, G. (2017).
\newblock A semiparametric and location-shift copula-based mixture model.
\newblock {\em Journal of Classification}, 34(3):444--464.

\bibitem[\protect\astroncite{Mazo and Averyanov}{2019}]{mazo2019constraining}
Mazo, G. and Averyanov, Y. (2019).
\newblock Constraining kernel estimators in semiparametric copula mixture
  models.
\newblock {\em Computational Statistics \& Data Analysis}, 138:170--189.

\bibitem[\protect\astroncite{Meyer}{1976}]{meyer1976sufficient}
Meyer, R.~R. (1976).
\newblock Sufficient conditions for the convergence of monotonic mathematical
  programming algorithms.
\newblock {\em Journal of Computer and System Sciences}, 12:108--121.

\bibitem[\protect\astroncite{Nelsen}{2007}]{nelsen2007introduction}
Nelsen, R.~B. (2007).
\newblock {\em An introduction to copulas}.
\newblock Springer Science \& Business Media.

\bibitem[\protect\astroncite{Qiang}{2010}]{qiang_high-order_2010}
Qiang, J. (2010).
\newblock A high-order fast method for computing convolution integral with
  smooth kernel.
\newblock {\em Computer Physics Communications}, 181(2):313--316.

\bibitem[\protect\astroncite{Rau et~al.}{2015}]{rau2015co}
Rau, A., Maugis-Rabusseau, C., Martin-Magniette, M.-L., and Celeux, G. (2015).
\newblock Co-expression analysis of high-throughput transcriptome sequencing
  data with {P}oisson mixture models.
\newblock {\em Bioinformatics}, 31(9):1420--1427.

\bibitem[\protect\astroncite{Wu and Lange}{2010}]{wu2010mm}
Wu, T.~T. and Lange, K. (2010).
\newblock The {MM} alternative to {EM}.
\newblock {\em Statistical Science}, 25(4):492--505.

\bibitem[\protect\astroncite{Xiang et~al.}{2019}]{xiang_overview_2019}
Xiang, S., Yao, W., and Yang, G. (2019).
\newblock An {Overview} of {Semiparametric} {Extensions} of {Finite} {Mixture}
  {Models}.
\newblock {\em Statistical Science}, 34(3):391--404.
\newblock Publisher: Institute of Mathematical Statistics.

\bibitem[\protect\astroncite{Zangwill}{1969}]{zangwill1969book}
Zangwill, W.~I. (1969).
\newblock {\em Nonlinear Programming---A Unified Approach}.
\newblock Prentice-Hall.

\end{thebibliography}
\bibliographystyle{apa}

\end{document}